\documentclass[aps,prl,reprint,superscriptaddress]{revtex4-2}
\usepackage{cmap}
\usepackage{amsmath, amsfonts, amsthm}
\usepackage{tikz}
\usepackage{subcaption}
\usepackage[hidelinks]{hyperref}
\usepackage{adjustbox}
\usepackage[all,cmtip]{xy}
\usepackage{dcolumn}
\usepackage{multirow}
\usepackage{amssymb}

\newtheorem{theorem}{Theorem}
\newtheorem{corollary}{Corollary}
\newtheorem{conjecture}{Conjecture}
\newtheorem{example}{Example}
\newtheorem{prop}{Proposition}
\newenvironment{psmallmatrix}
  {\left(\begin{smallmatrix}}
  {\end{smallmatrix}\right)}

\newcommand{\drawgenerator}[8]{%
\xymatrix@!0{%
& #8 \ar@{-}[ld]\ar@{.}[dd] \ar@{-}[rr] & & #7 \ar@{-}[ld]  \\%
#1 \ar@{-}[rr] \ar@{-}[dd] &  & #2 \ar@{-}[dd] &            \\%
& #6 \ar@{.}[ld] &  & #5 \ar@{-}[uu] \ar@{.}[ll]       \\%
#3 \ar@{-}[rr] &  & #4 \ar@{-}[ru]                       %
}}

\newcommand{\plaquette}[4]{
\xymatrix@!0{%
#1 \ar@{-}[r] \ar@{-}[d]  & #2 \ar@{-}[d] 
\\
#3 \ar@{-}[r]  & #4
}}

\begin{document}

\title{
Translation-invariant quantum low-density parity-check codes from compactified fracton models
}

    \author{Cassandra M.~Hopkin}
\affiliation{University of California, Davis, CA 95616, USA}
    \author{Victor V.~Albert} 
\affiliation{Joint Center for Quantum Information and Computer Science, NIST/University of Maryland, College Park, MD 20742, USA}
\author{Dominic J.~Williamson}
\affiliation{School of Physics, The University of Sydney, NSW 2006, Australia}

\begin{abstract}
Quantum error-correcting codes with translation symmetry and local checks have been studied extensively, leading to a wide variety of fracton codes in three or more dimensions which lack a complete unifying picture. Recently, the study of translation-invariant codes with long-range checks has revealed impressive performance for small fixed-size instances in two dimensions. Here, we provide a unifying picture for a large family of translation-invariant codes, both local and long-range, that captures many fracton codes and all Abelian Two-Block Group Algebra (A2BGA) codes, including the Bivariate Bicycle (BB) codes. The balanced product structure of A2BGA codes leads to a local parent code that is a hypergraph product fracton model in a higher dimension. Different compactifications of a parent code produce a wide variety of descendant codes which provides a unifying picture for their properties. In particular, all BB codes with the same check weight are derived from a single parent hypergraph product fracton model. This construction allows us to extend Wang and Pryadko's code-parameter bounds for Generalized Bicycle codes to A2BGA codes. We conjecture that the transversal gates and energy barriers of the translation-invariant descendant codes are limited by those of their parent fracton models.
\end{abstract}

\maketitle

% ////// INTRODUCTION //////////////////////////
\section{Introduction}

Errors occur when quantum information is stored, processed, and transmitted due to decoherence from uncontrolled coupling to the environment.
To maintain coherence it is necessary to construct fault-tolerant protocols that detect and correct quantum errors. 
Quantum error-correcting codes were created for this purpose. 
Curiously, these codes have other applications, including as models for topological phases of matter \cite{Kitaev_2003}. 

Quantum error-correcting (QEC) codes with local checks on lattices in three dimensions, or higher, can support exotic fracton phases of matter that are characterized by topological excitations with constrained mobility~\cite{Chamon_2005}. 
The classification of fracton codes poses an interesting challenge for theoretical condensed matter physics~\cite{sorting}. Fracton codes can also support physical properties that are advantageous for the purposes of QEC. In particular, the cubic code has no stringlike logical operators and can be used to implement a partially self-correcting quantum memory \cite{cubicCode, Bravyi_2013}. 
Fracton models support a diverse array of topological excitation mobility structures, coarsely grouped into Type-I and Type-II~\cite{fractonReview}. A complete unifying framework for the classification of fracton models remains an open problem, despite recent progress~\cite{Song_2023, Aasen_2020, Pai_2019, Wickenden_2025, Wickenden_Shirley_2025, Song_2024, Slagle_2019}. 

Local fracton models in three dimensions and higher can also come about from the hypergraph product (HGP) construction (see Ref.~\cite{fractonProduct}), which combines two classical codes to make a quantum code~\cite{ogHGP}.  
Similarly, in Ref.~\cite{orthoplex} codes called \textit{orthoplex models} that exhibit fracton order are constructed via a generalized HGP construction. 

In this work, we combine the HGP construction with twisted boundary conditions (\textit{a.k.a.} compactification) to produce a wide variety of fracton models from a single parent higher dimensional HGP fracton code~\cite{compactifyingFractons}. 
In addition, we show that families of translation-invariant quantum low-density parity-check (qLDPC) codes with checks of fixed weight but potentially growing range are also compactifications of a parent fracton code that has local translation-invariant checks on a higher-dimensional lattice. 
Our analysis is not restricted to bivariate-bicycle (BB) codes with checks of a fixed form (cf.~\cite{yu-an,coveringGraphs}) which allows us to relate larger families of codes with growing check range. 
In particular, all Abelian Two-Block Group-Algebra (A2BGA) codes~\cite{twoBlock, twoBlock2} with the same fixed check weight on left and right qubits are compactifications of a single higher-dimensional HGP fracton code. 
This includes all BB codes~\cite{Panteleev_2021_GB, BBCodes}, all CSS cubic codes~\cite{cubicCode}, and all fractal spin-liquid codes~\cite{fractalSpinLiquids}. 
The parent fracton code is derived from a HGP of simple translation-invariant classical fractal codes~\cite{PhysRevE.60.5068}. 

To derive our result, we employ Haah's polynomial formalism to describe the HGP and compactification via the imposition of twisted boundary conditions~\cite{commuting_pauli_hamiltonians}. 
Our result establishes that families of BB codes are local in a dimension that depends on the form of their checks. This dimension is between two, for fixed checks, and $(w-2)$ for increasingly long-range checks, where $w$ is the check weight of the BB code family.

We characterize the class of models that can be derived from the same high-dimensional HGP code as a \textit{fracton family tree}. We discuss three fracton family trees that arise for qubits, corresponding to polynomials over the field $\mathbb{F}_2$, and list notable examples of codes that belong to each family. 
In particular, we find that families of BB codes with increasingly long-range checks of a fixed weight are encompassed by these family trees. 

We extend results from Refs.~\cite{BBdistBounds} to establish distance bounds for the compactified codes, including those with increasingly long-range checks. 
These bounds are derived from the higher-dimensional local parent codes~\cite{BPTBound,BBdistBounds}, and we also discuss stricter bounds for codes with fixed check locality~\cite{yu-an}. 
We note that the QLDPC upper bounds of Ref.~\cite{checkQLDPCBounds} apply more generally than ours since they do not rely on translation invariance; however, as a consequence, they are not as tight. 
We then discuss potential logical gate restrictions on the compactified models following the results of Ref.~\cite{noGoHGP}. 
Finally, we propose a conjecture about the connection between the energy barrier of the parent code and its descendants. 

There are several related code constructions that involve higher dimensional mappings, which we now discuss. In Ref.~\cite{Hastings_manifolds} Hastings considered a family of generalized toric codes defined on high dimensional manifolds, where the dimension of the manifold increases with the number of qubits. It was shown that, assuming a geometric conjecture, this family of codes achieve logarithmic weight stabilizers and almost-linear distance scaling. 
Freedman and Hastings~\cite{FreedmanHastings} introduced a construction that takes in an arbitrary CSS code and outputs an 11-dimensional generalized surface code. Despite mapping qLDPC codes to higher dimensional topological codes, this construction is distinct from the one presented here. 
Another distinct line of work on higher dimensional HGP codes focuses on projecting hypergraph products of codes on fractal lattices into three dimensions~\cite{coredProduct, Brell_2016, Lin_2024}. 

The work is laid out as follows. 
In Section~\ref{sec:Background} we introduce relevant background material. 
In Section~\ref{sec:Compactification} we establish our central result connecting translation-invariant qLDPC codes and compactified fracton codes. 
In Section~\ref{sec:Distance} we discuss the distance of compactified HGP fracton codes. 
In Section~\ref{sec:logicalGates} we discuss restrictions on fault-tolerant logic gates for compactified HGP fracton codes. 
In Section~\ref{sec:Energy} we discuss the energy barriers of parent HGP fracton codes and their compactifications.
In Section~\ref{sec:Discussion} we conclude with a discussion of our results and future directions.

% ////// BACKGROUND ///////////////

\section{Background}
\label{sec:Background}

In this section we provide background on classical and quantum error correction, the notation we use for translation-invariant codes, and specific code families of interest.

\subsection{Classical Error Correction}
\label{subsec:classical}

A binary \textit{classical code} is a set of vectors which we call \textit{codewords}. The code is linear if any linear combination of codewords is also a codeword. If this is the case, the code can be specified using a \textit{parity-check matrix} $H$, where the codewords are vectors in the kernel of $H$. A familiar example is the [7,4,3]-Hamming code, which has the following parity-check matrix:
\begin{equation}
H = \begin{pmatrix}
    1 & 1 & 0 & 1 & 1 & 0 & 0 \\
    1 & 0 & 1 & 1 & 0 & 1 & 0 \\
    0 & 1 & 1 & 1 & 0 & 0 & 1
\end{pmatrix}.
\end{equation}

Observe that the matrix specifies the parity checks that the codewords need to satisfy in each row. Next, we define the quantum analogue of linear codes which can also be defined using a parity-check matrix $H$. 

\subsection{Stabilizer Formalism}
\label{subsec:stablizers}

\textit{Stabilizer codes}~\cite{gottesmanStab,calderbank1997quantum} are quantum analogues of linear codes. The stabilizer formalism defines quantum codes in terms of operators, as opposed to directly defining the codewords. We begin with the Pauli group $\cal{P}$, defined on $n$ qubits as the set of all possible Pauli strings with fourth-root-of-unity phases,
\begin{equation}
    \mathcal{P} = \{\sigma_1 \otimes \sigma_2 \otimes \cdot \cdot \cdot \otimes \sigma_n\},
\end{equation}
where \[\sigma_i \in \{\pm I, \pm iI, \pm X, \pm iX, \pm Y, \pm iY, \pm Z, \pm iZ\},\] where $X, Y, Z$ are the Pauli operators, and where $I$ is the identity. Then the \textit{stabilizer group} is an abelian subgroup of the Pauli group that does not include $-I$. This ensures that each stabilizer has $+1$ as an eigenvalue, and therefore the stabilizers have a joint +1 eigenspace. The codewords of the stabilizer code are then defined to be the vectors in this joint eigenspace.

All of the codes that we consider in this work have a stabilizer group that separates into two subsets, one with stabilizers with only Pauli $X$ terms and one with stabilizers with only Pauli $Z$ terms. Stabilizer codes of this form are called \textit{Calderbank-Shor-Steane} (CSS) codes~\cite{calderbank1996good,steane1996error,steane1996multiple}. We can equivalently think of CSS codes as combinations of two linear classical codes, where one code corrects for $X$-type errors, and another corrects for $Z$-type errors. However, to satisfy the requirement that all of the stabilizers commute, the classical codes must be orthogonal. This means that if we have two classical linear codes $C_1$ and $C_2$ with parity-check matrices $H_X$ and $H_Z$ respectively, they must satisfy
\begin{equation}
    H_XH_Z^{T} = 0.
\end{equation}
We can avoid this restriction by using the hypergraph product construction, which guarantees the orthogonality of $H_X$ and $H_Z$ for arbitrary classical linear codes.

% ////// HYPERGRAPH PRODUCT ///////////////

\subsection{Hypergraph Product Codes}
\label{subsec:HGP}
Suppose we have two classical linear codes with parity-check matrices $H_1$ and $H_2$. Then we define their \textit{hypergraph product} by the following equations. 
\begin{equation}
    H_X = [H_1 \otimes I \mid I \otimes H_2] \qquad H_Z = [I \otimes H_2^T \mid H_1^T \otimes I].
\end{equation}
We can check that this is a valid CSS code by verifying 
\begin{equation}
    H_XH_Z^{T} = H_1 \otimes H_2 + H_1 \otimes H_2 = 0
\end{equation}
since we are working in $\mathbb{F}_2$. Then the quantum parity-check matrix is written as follows: 
\begin{equation}
    H = \begin{pmatrix}
    0 & H_Z \\ 
    H_X & 0
\end{pmatrix}.
\end{equation}

In this work we consider the hypergraph product of translation-invariant classical codes defined on a hypercubic lattice. We associate the physical bits with vertices on the lattice, and specify the parity-checks by shaded regions connecting bits. We can then write the parity-check matrix by associating a basis vector to each bit. We can define a whole family of translation-invariant classical codes by considering the checks on an infinite lattice, and then changing the specific boundary conditions. The matrix notation becomes cumbersome to represent these classical families of codes, since the parity-check matrix can change drastically depending on the specific boundary conditions. As a result, we turn to the compact polynomial formalism used by Haah \cite{haahThesis}.

% ////// POLYNOMIAL FORMALISM ///////////////

\subsection{Polynomial Formalism}
\label{subsec:poly}
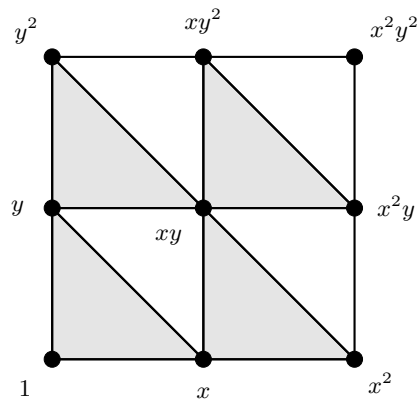
\begin{figure}
    \centering
    \begin{center}
    \begin{tikzpicture}
    \draw[thick] (0,0) -- (4,0) -- (4,4) -- (0,4) -- (0,0);
    \draw[thick] (2,4) -- (2,0);
    \draw[thick] (0,2) -- (4,2);
    \draw[thick] (0,4) -- (4,0);
    \draw[thick] (0,2) -- (2,0);
    \draw[thick] (2,4) -- (4,2);
    \draw [fill = gray, fill opacity=0.2, thick] (0,0) -- (0,2) -- (2,0) -- cycle;
    \draw [fill = gray, fill opacity=0.2, thick] (0,4) -- (0,2) -- (2,2) -- cycle;
    \draw [fill = gray, fill opacity=0.2, thick] (2,4) -- (2,2) -- (4,2) -- cycle;
    \draw [fill = gray, fill opacity=0.2, thick] (2,2) -- (2,0) -- (4,0) -- cycle;
    \filldraw[black] (0,0) circle (3pt)
      node[label={[label distance=2pt]below left:{$1$}}] {};
    \filldraw[black] (2,0) circle (3pt)
      node[label={[label distance=4pt]below:{$x$}}] {};
    \filldraw[black] (4,0) circle (3pt)
      node[label={[label distance=-2pt]below right:{$x^2$}}] {};
    \filldraw[black] (0,2) circle (3pt)
      node[label={[label distance=4pt]left:{$y$}}] {};
    \filldraw[black] (0,4) circle (3pt)
      node[label={[label distance=-2pt]above left:{$y^2$}}] {};
    \filldraw[black] (2,2) circle (3pt)
      node[label={[label distance=2pt]below left:{$xy$}}] {};
    \filldraw[black] (4,2) circle (3pt)
      node[label={[label distance=2pt]right:{$x^2y$}}] {};
    \filldraw[black] (2,4) circle (3pt)
      node[label={[label distance=2pt]above:{$xy^2$}}] {};
    \filldraw[black] (4,4) circle (3pt)
      node[label={[label distance=-1pt]above right:{$x^2y^2$}}] {};
  \end{tikzpicture}
\end{center}
    \caption{A classical code generated by the polynomial $1 + x + y$ (Newman-Moore Code). Here, the site with coordinates $jk$ is indexed by $x^j y^k$. The circles depict bits and the shaded triangles depict three body parity checks.}
    \label{lattice}
\end{figure}

We associate a monomial with each possible lattice coordinate such that multiplication by each variable corresponds to moving in a unique direction. For example, in Figure \ref{lattice}, we move through the lattice to the right by multiplying by $x$, and move up through the lattice by multiplying by $y$. We can formally define a family of codes by viewing the checks on an infinite lattice, and recover specific codes that belong to this family by imposing boundary conditions on the variables. These codes have checks on a finite lattice, and can be defined using cyclic permutation matrices, as described below.

Consider a classical code $\mathcal{C}$ defined on an $l \times m$ two-dimensional lattice such that $x^l = y^m = 1$. Then the correspondence with the matrix notation is given by
\begin{equation}
    x = S_l \otimes I_m \qquad y = I_l \otimes S_m    ,
    \label{xydef}
\end{equation}
where $S_j$ is an $j \times j$ cyclic permutation matrix such that $(S_j)^j = I_j$. It is easy to see that $x$ and $y$ commute since they have non-identity terms in different tensor factors. 

This definition of the variables in terms of cyclic permutation matrices can be readily generalized to higher dimensions. The \textit{generating polynomial} of a classical code is defined as the sum of the monomials involved with one of the translation-invariant checks. Note that this polynomial is not unique, as we can multiply by any monomial to get an equivalent generator. By convention, we usually choose the generating polynomial to be associated with the check nearest the origin. Then the parity-check matrix can be represented in polynomial notation by all of the shifts of the generator matrix. To illustrate this, the classical code in Figure \ref{lattice} has parity check matrix 
\begin{equation}
    H = x^{n_1}y^{n_2}(1 + x + y),
\end{equation}
for all $n_{1,2}$.
Notice that this polynomial description of the parity check matrix is the same for any of the codes in the same family, where specific codes are defined by conditions on the variables $x$ and $y$. 
\newline

We can then use this to rewrite the HGP construction using the polynomial formalism. Suppose we have two classical codes that are $m$ and $n$-dimensional with generating matrices $f(x_1,...,x_m)$ and $g(y_1,...,y_n)$ respectively. For compactness of notation, let 
\begin{equation}
    f(x_1,...,x_m) = f(\mathbf{x}) \quad \text{and} \quad g(y_1,...,y_n) = g(\mathbf{y}),
\end{equation}
and denote 
\begin{equation}
    \mathbf{x^e} = x_1^{e_1}\cdots x_m^{e_m} \qquad \mathbf{y^r} = y_1^{r_1}\cdots y_n^{r_n}.
\end{equation}
Following the definition given in the previous section, we can write 
\begin{subequations}
\begin{align}
    H_X &= [\mathbf{x^e}f(\mathbf{x}) \otimes I_{2} \mid I_{1} \otimes \mathbf{y^r}g(\mathbf{y})] \\
    H_Z &= [I_1 \otimes \mathbf{y^{-r}}\overline{g(\mathbf{y})} \mid \mathbf{x^{-e}}\overline{f(\mathbf{x})} \otimes I_2]
\end{align}
\end{subequations}
where the overline represents taking the inverse of each monomial in the polynomial. We observe that we can write $I_{1} = x_1^{e_1}...x_m^{e_m}(1)$ and $I_{2} = y_1^{r_1}...y_n^{r_n}(1)$. Then we can recast $H_X$ and $H_Z$ as follows:
\begin{subequations}
\begin{align}
    H_X &= [\mathbf{x^e}\mathbf{y^r}f(\mathbf{x}) \mid \mathbf{x^e}\mathbf{y^r}g(\mathbf{y})] \\
    H_Z &= [\mathbf{x^{-e}}\mathbf{y^{-r}}\overline{g(\mathbf{y})} \mid \mathbf{x^{-e}}\mathbf{y^{-r}}\overline{f(\mathbf{x})}].
\end{align}
\end{subequations}

Since we assume translation-invariance for the input classical codes, we can omit the shift terms and write the checks using the more compact notation given below. We note that this form is similar to the notation used by Yoshida for fractal spin liquids \cite{fractalSpinLiquids} 
\begin{equation}
    X\begin{pmatrix}
    f(\mathbf{x}) \\ g(\mathbf{y})
\end{pmatrix} \qquad Z\begin{pmatrix}
    \overline{g(\mathbf{y})} \\ \overline{f(\mathbf{x})}
\end{pmatrix}.
\label{eq:hgp_notation}
\end{equation}

Although the HGP construction only uses two generating polynomials, in general we can use arbitrarily many. In particular, $q$ generating Laurent polynomials in $D$ formal variables corresponds to a code defined in $D$ dimensions with $q$ qubits per site \cite{Haah_Lecture_Notes}.

% ////// CODES THAT EXHIBIT TOPOLOGICAL ORDER ///////////////

\subsection{Topological Codes}
\label{subsec:topological}

Next, we discuss a surprising connection between quantum error correction and condensed matter physics, using the toric code as an illustrative example~\cite{Kitaev_2003}. 
This motivates our exploration of exotic topological phases which can be modelled by quantum error correcting codes. 

We say that a system exhibits \textit{topological order} when any geometrically local operator that commutes with all stabilizers must itself be a product of stabilizers \cite{topoCondAlgo}. This implies that the codewords, which correspond to the ground states of the system in the condensed matter perspective, exhibit local indistinguishability.

The prototypical example of a QEC code that exhibits topological order is the toric code. This code can be defined on a 2D lattice with two qubits per vertex, and its stabilizer generators (\textit{a.k.a.} check operators, or check for short) are
%toric code
\begin{eqnarray}
\label{toricCode}
\begin{array}{c}
\plaquette{IX}{II}{XX}{XI}
\qquad\plaquette{IZ}{ZZ}{II}{ZI}
\end{array}
\, .
\label{unit_cell2D}
\end{eqnarray}
An equivalent definition in the polynomial formalism is
\begin{equation}
    X\begin{pmatrix}
    1 + x \\ 1 + y
\end{pmatrix} \qquad Z\begin{pmatrix}
    x + xy \\ y + xy
\end{pmatrix}.
\label{toricpoly}
\end{equation}
We can use these stabilizers to define the interaction terms of the following Hamiltonian: 
\begin{equation}
    \label{hamiltonian}
   H = -\sum_{j, k} \Big[ x^j y^k\begin{pmatrix}
       1 + x \\
       1 + y
   \end{pmatrix} + x^j y^k\begin{pmatrix}
       x + xy \\
       y + xy
   \end{pmatrix} \Big]
\end{equation}
The ground states of the Hamiltonian in Eq.~\ref{hamiltonian} correspond to the codewords of the toric code, and the pairs of excitations correspond to errors.

In this context, we define \textit{excitations} to be -1 eigenvalues of the stabilizers, that is, syndromes that correspond to an error. When viewed on a lattice, these excitations correspond to vertices on the boundary of an error path. A \textit{string operator} (or \textit{string segment}) is a Pauli operator with finite support whose excitations are contained in two boxes, each with finite width of linear size $w$ \cite[Definition 5.1]{haahThesis}. We define the length $l$ of the string operator to be the distance between the two boxes. Intuitively, this is an operator that can extend infinitely in one direction, while its width stays constant in the remaining directions (when defined on an infinite lattice). Some of the constructions we describe require the \textit{no strings} rule, which means that $l \geq \alpha w$ where $l$ is the length of the string and $\alpha$ is some constant, $\alpha \geq 1$.

% ////// FRACTON MODELS ///////////////

\subsection{Fracton Models}
\label{subsec:fracton}

We are interested in classifying \textit{fracton} models, which are local topological codes that model exotic fracton phases of matter. These phases of matter are characterized by having immobile isolated excitations, and composite excitations that are potentially mobile. As a result, each individual excitation is defined as having ``fractionalized mobility" \cite{fractonReview}. We consider fracton models that arise from exactly solvable lattice models in three, or more, dimensions. Within these, there are two types of models. Type-I models have fracton excitations, and composite excitations are only mobile in lower-dimensions. Type-II fracton models have no non-trivial mobile excitations \cite{fractonReview, fractonTypes, cubicCode}. 
This QEC formulation of fracton models allows for the application of quantum information theoretic ideas to further research in condensed matter physics, such as finding new topological phases of matter or studying quantum glassiness~\cite{Williamson_2016, Kim_2016, Prem_2017, Prem_2019, Williamson_2023, Bulmash_2019, Prem_2019_permutation_symmetries, Gorantla_2025, Devakul_2021}. 
Fracton models have implications in quantum information, such as for creating robust \textit{partially self-correcting} quantum memories~\cite{cubicCode, Michnicki_2014, Siva_2017, Williamson_2024, williamson2025partialselfcorrection, Gu2025}
and new decoders~\cite{Brown_2020, Schwartzman_Nowik_2025, Newman-MooreEnergyBarrier}. 

The most prominent fracton memory is Haah's cubic code \cite{cubicCode}. This code is a three-dimensional Type-II fracton model with two qubits per lattice site.
Equation~\ref{cubic1} shows how multiplication by variables is equivalent to translation in the polynomial formalism, while the checks of the code are depicted in Eq.~\ref{cubic2}.
%Cubic code
\begin{eqnarray}
\label{cubic1}
\begin{array}{c}
\drawgenerator{xz}{xyz}{x}{xy}{y}{1}{yz}{z}
\end{array}
\, .
\label{unit_cell}
\end{eqnarray}
\begin{align}
\label{cubic2}
\begin{array}{c}
\drawgenerator{XI}{II}{IX}{XI}{IX}{XX}{XI}{IX}
\qquad
\drawgenerator{ZI}{ZZ}{IZ}{ZI}{IZ}{II}{ZI}{IZ}
\end{array}
\, .
\end{align}
For example, the left figure of Eq.~\ref{cubic2} shows that the $X$-checks are defined on the first qubit by the polynomial $1 + xy + xz + yz$ and by $1 + x + y + z$ on the second qubit. 
Thus in the polynomial formalism, Haah's code has checks given by 
\[X\begin{pmatrix}
    1 + x + y + z \\
    1 + xy + xz + yz
\end{pmatrix} \qquad Z\begin{pmatrix}
    x + y + z + xyz \\
    xy + xz + yz + xyz
\end{pmatrix}.\] We observe that this code is not itself a hypergraph product, as the same variables appear in both the polynomials $f$ and $g$.

The motivating problem for this work is the diversity of fracton models and the lack of underlying principles that connect them. 
To resolve this, we use high-dimensional hypergraph products to encapsulate families of CSS fracton models. Notably, these parent codes incorporate both Type-I and Type-II models into the same family.

% ////// BB Codes ///////////////

\subsection{Bivariate Bicycle Codes}
\label{subsec:BB}

Using the definition of $x$ and $y$ given in Eq.~\eqref{xydef}, a BB code is defined by block matrices 
\begin{equation}
    A = A_1 + A_2 + A_3 \qquad B = B_1 + B_2 + B_3,
\end{equation}
where each $A_i$ and $B_i$ is a power of $x$ or $y$. Then we define 
\begin{equation}
 H_X = [A | B] \qquad H_Z = [B^{T} | A^{T}]. 
\end{equation}

An example of a BB code that is promising for near-term quantum computing is the \textit{gross code}, which is defined by matrices
\begin{equation}
    A = x^3 + y + y^2 \qquad B = y^3 + x + x^2.
\end{equation}
This code is promising for near-term experiments due to its low-weight checks, large distance ($d = 12$), and high ancilla-added~\cite{BBCodes} encoding rate ($r = 1/24$).

The BB codes fall into the more general category of Abelian Two-Block Group Algebra (A2BGA) codes, where the generating polynomials are not restricted to having exactly two variables and three terms, and are allowed to have mixed terms \cite{twoBlock}. 
Each polynomial is allowed to depend on the same set of variables, so these are not directly HGPs of two classical codes as in Eq.~\eqref{eq:hgp_notation}.
Our key observation is that the two code families are related by imposition of certain relations among the distinct variable sets of the HGP construction.

% ////// UNIFYING CODES ///////////////

\section{Unifying Translation-Invariant Codes}

\label{sec:Compactification}

In this section we discuss our method of mapping translation-invariant A2BGA codes to higher dimensional local HGP codes. We begin in Section \ref{subsec:compact} with a discussion of compactification, which allows us to recover lower dimensional codes from a high dimensional parent. Next, in Section \ref{subsec:verify} we discuss an algebraic criterion for when a code is indecomposable, since our unification method only works for these codes. In Section \ref{subsec:unification}, we present our method for unifying A2BGA codes into families with a common HGP parent. 
We conclude in Section \ref{subsec:families} with an explicit description of these families when working in $\mathbb{F}_2$.

\subsection{Compactification}
\label{subsec:compact}
We now introduce a method of dimension reduction to go from a high-dimensional HGP parent to a lower-dimensional descendant code. Suppose we have a family of models that all have the same translation-invariant stabilizers and varying size periodic boundary conditions. This means that the codes have the same parity-check matrix in the polynomial notation. Equations of the form $x_i^{l_i} = 1$ fix the size of the lattice which the code is defined on, and we can scale the code up or down by adjusting these equations. Then a higher-dimensional code can give rise to a lower-dimensional code if additional twisted boundary conditions are imposed, and we call this process \textit{compactification}~\cite{compactifyingFractons}. 

We impose twisted conditions by enforcing an equation that sets a product of the variables $x_i$ and their inverses to one, i.e. 
\begin{align}
    \prod_{i}x_i^{v_i} = 1
\end{align}
where $v_i$ are integers. 
Figure \ref{twisted} shows an example compactification of a 2D lattice to 1D, given by the equation $x^3 = y$. The compactified codes can alternatively be viewed as a balanced product, where each code arises by quotienting the HGP code by the polynomial representing the twisted conditions \cite{balancedProd}. 

In order to ensure that the lower-dimensional code has growing distance with the system size, it needs to have no local logical operators. This means that the parent model should not have logical operators that wrap around the compactified directions without spreading out in the other directions. In three dimensions, this corresponds to the code having no string-like operators in the compactified direction.

We note that while we consider compactifications of high dimensional HGP codes in this work, lower-dimensional fracton models can also be studied through properties of their compactifications. Specifically, Dua et al. study two-dimensional compactifications of three dimensional fracton models, finding differences in the topological order of the compactifications based on whether the parent code was Type-I or Type-II \cite{compactifyingFractons}. The topological mapping of BB codes from Ref.~\cite{yu-an} can also be understood through the lens of compactification. In this case a fixed compactification to a two dimensional code is used, along with potentially twisted boundary conditions on the two-dimensional code. 

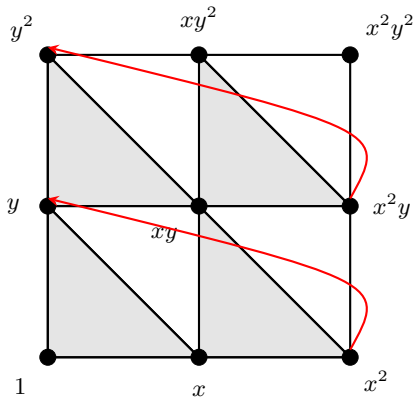
\begin{figure}
    \centering
    \begin{tikzpicture}
    \draw[thick] (0,0) -- (4,0) -- (4,4) -- (0,4) -- (0,0);
    \draw[thick] (0,2) -- (4,2);
    \draw[thick] (0,4) -- (4,0);
    \draw[thick] (0,2) -- (2,0);
    \draw[thick] (2,4) -- (4,2);
    \draw [fill = gray, fill opacity=0.2, thick] (0,0) -- (0,2) -- (2,0) -- cycle;
    \draw [fill = gray, fill opacity=0.2, thick] (0,4) -- (0,2) -- (2,2) -- cycle;
    \draw [fill = gray, fill opacity=0.2, thick] (2,4) -- (2,2) -- (4,2) -- cycle;
    \draw [fill = gray, fill opacity=0.2, thick] (2,2) -- (2,0) -- (4,0) -- cycle;
    \filldraw[black] (0,0) circle (3pt)
      node[label={[label distance=2pt]below left:{$1$}}] {};
    \filldraw[black] (2,0) circle (3pt)
      node[label={[label distance=4pt]below:{$x$}}] {};
    \filldraw[black] (4,0) circle (3pt)
      node[label={[label distance=-2pt]below right:{$x^2$}}] {};
    \filldraw[black] (0,2) circle (3pt)
      node[label={[label distance=4pt]left:{$y$}}] {};
    \filldraw[black] (0,4) circle (3pt)
      node[label={[label distance=-2pt]above left:{$y^2$}}] {};
    \filldraw[black] (2,2) circle (3pt)
      node[label={[label distance=2pt]below left:{$xy$}}] {};
    \filldraw[black] (4,2) circle (3pt)
      node[label={[label distance=2pt]right:{$x^2y$}}] {};
    \filldraw[black] (2,4) circle (3pt)
      node[label={[label distance=2pt]above:{$xy^2$}}] {};
    \filldraw[black] (4,4) circle (3pt)
      node[label={[label distance=-1pt]above right:{$x^2y^2$}}] {};
    \draw[->, thick, red, >=stealth] 
    (4,0.1) .. controls (4.5,1) and (4.5,1) .. (0,2.1);

    \draw[->, thick, red, >=stealth]
        (4,2.1) .. controls (4.5,3) and (4.5,3) .. (0,4.1);
    
    \end{tikzpicture}
    \caption{
    Twisted boundary conditions on a 2D lattice given by the equation $x^3 = y$. We call the general procedure of applying twisted boundary conditions ``compactification''; it can be thought of as a special case of a balanced product \cite[Fig. 7]{balancedProd}.
    }
    \label{twisted}
\end{figure}

\subsection{Required Model Properties}
\label{subsec:verify}

In order for the unification method described in the next section to work, we need to first ensure that the original code is indecomposable. We say that a code is \textit{decomposable} if there exists a choice of generating stabilizers that are supported on two distinct subsets of qubits. In the context of translation-invariant codes, decomposable codes decouple into codes on multiple sublattices.
We next present an algebraic criterion for whether an A2BGA code is indecomposable. 

Suppose we have a model generated by the polynomials $f$ and $g$ with variables $\{x_i\}_{i=1,\dots,d}$. Consider the group $G$ generated by the monomials that appear in $f$ and $g$ (along with their inverses and 1). We call this the \textit{monomial group}. Then we have the following necessary and sufficient condition for whether a model is decomposable. 
\begin{theorem}
     An A2BGA code $\mathcal{C}$ with monomial group $G$ is indecomposable if and only if $G=\langle \{x_i\}_{i=1,\dots,d} \rangle$. 
\end{theorem}

\begin{proof}
We first note that a code is decomposable if and only if there exists a choice of stabilizer generators such that the Tanner graph is disconnected. This is because the Tanner graph consists of multiple connected components exactly when there exists a choice of stabilizer generators that can be divided into two or more sets which each have support on distinct sets of qubits. Therefore, a code is indecomposable if any choice of stabilizer generators produces a connected Tanner graph. 

We now show that a given Tanner graph of $\mathcal{C}$ is connected if and only if the monomial group condition holds. Let $\mathcal{C}$ be generated by polynomials $f, g$ and let $m_i$ denote the monomials of $f$, and $n_j$ denote monomials of $g$. Also assume that $f$ and $g$ are of the form $1 + ...$ so that $1 \in \{m_i\}$ and $1 \in \{n_j\}$. Similarly to Lemma 3 in Ref \cite{BBCodes}, consider the set $S = \{m_i m_j^{-1}\} \cup \{n_k n_l^{-1}\}$ where $i, j$ are any term index in $f$ and $k, l$ are any term index in $g$. Suppose $S$ generates the entire lattice. Let $\alpha$ be any left qubit in the lattice. Then $\alpha$ is connected in the Tanner graph of $\mathcal{C}$ to some left qubit $\alpha \cdot m_i \cdot m_j^{-1}$ through an $X$ check node. Since we supposed that $S$ generates the entire lattice, $\alpha$ is connected to any other left qubit in the lattice. Similarly, all of the right qubits are connected through $Z$ check nodes. We know that each $X$ check and $Z$ check is connected to a left and right qubit, so the Tanner graph must be connected. Since paths between nodes in the Tanner graph are defined as multiplications by elements of $S$, if the Tanner graph is connected then $S$ generates the entire lattice. 

We now show that $S = G$. Suppose that $m \in G$ and that $m$ appears as a monomial in $f$ or $g$ (not its inverse). Then $m \cdot 1^{-1} \in S$. If $m$ is the inverse of a monomial appearing in $f$ or $g$, then $1 \cdot (m^{-1})^{-1} \in S$. Conversely, we know that any element in $S$ is a product of monomials appearing in $f$ or $g$ and must lie in $G$. Therefore, the Tanner graph is connected if and only if $G$ generates the entire lattice. Algebraically, this means that $G$ must be generated by the unique variables (\(x\), \(y\), \(z\), etc.) considered as degree-one monomials.

\end{proof}

To check whether a model meets the above condition, we can use the following steps. 
Suppose some variable $v$ in $f$ and $g$ appears as a degree-one monomial term. Then we can ignore all instances of $v$, and examine the polynomials without it. By repeating this process, suppose there are no variables that appear as degree-one monomials in the two polynomials. If we are left with two constant polynomials, then this means that each unique variable appeared as a degree-one monomial in $f$ and $g$, so we are done. If not, then we are left with sums of monomials that are products of the unique variables. Suppose that there do not exist two terms $\alpha$, $\beta$ that are powers of monomials in $f$ and $g$ that differ by exactly one degree-one variable, that is, $\alpha\beta^{-1} \neq u$ for all pairs of terms $(\alpha, \beta)$ and where $u$ is some unique variable that appears in $f$ and/or $g$. Then this means that there is no way to combine powers of the terms to get any of the remaining unique variables, so the model must be decomposable. For example,  consider
\begin{subequations}
\begin{align}
    f(x, y, z) &= 1 + x^{-1}y + x^{-1}z \\
    g(x, y, z) &= 1 + z^{-1}y~,
\end{align}
\end{subequations}
then the model has no checks involving terms of the form $x^i$, $y^i$, and $z^i$ where $i$ is odd. 

If we know that there exist two terms $\alpha, \beta$ that are powers of monomials in $f$ and $g$ that differ by only one degree-one monomial, then we can represent that variable (or its inverse) by $\alpha\beta^{-1}$. Then this means that we can ignore this variable, and inductively check if the condition holds for the resulting polynomials. If by repeating this process the polynomials become constant, then the monomial condition has been met.

\subsection{Unification via a parent model}
\label{subsec:unification}

Suppose we have an indecomposable fracton model of the form \[X\begin{pmatrix}
    f(\mathbf{x}) \\ g(\mathbf{x})
\end{pmatrix} \qquad Z\begin{pmatrix}
    \overline{g(\mathbf{x})} \\ \overline{f(\mathbf{x})}
\end{pmatrix}~.\]
Then without loss of generality, by the equivalence properties of translation invariant codes on infinite lattices described in Ref.~\cite{fractonEquiv}, we can write $f = 1 + f'$ and $g = 1 + g'$ for some other polynomials $f',g'$. 
Let $\mathbf{x} = (x_1,...,x_d)$, and define $d$ to be the number of unique variables across $f'$ and $g'$. Suppose that $f'$ and $g'$ have $k_1$ and $k_2$ monomial terms respectively, 
\begin{align}
    f'= m_1^{(f)} + m_2^{(f)}  + \cdots + m_{k_1}^{(f)} ,
    \\
    g'= m_1^{(g)} + m_2^{(g)} + \cdots + m_{k_1}^{(g)}.
\end{align}

First we consider the case where there are strictly more variables than monomial terms. In this case the model must decompose into decoupled models. This is because the lattice generated by the monomials is lower dimensional that the lattice generated by all the variables. We focus on codes that have already been reduced to an indecomposable form, and hence the number of variables is equal to or less than the number of monomial terms.

Now we consider the generic case of an indecomposible code on an infinite lattice where there are more terms than variables, that is, $d \leq k_1 + k_2$. Consider the HGP of the two classical codes generated by the following polynomials 
\begin{align}
    1 + a_1 + ... + a_{k_1}, \qquad 1 + b_1 + ... + b_{k_2},
\end{align}
where the $a_i$ and $b_j$ are newly introduced independent variables. 

Conceptually, we want to map $d$ of the $a_i$ and $b_i$ terms to the monomials in $f$ and $g$ that appear at the same index $i$. Since the monomial group of $f'$ and $g'$ is generated by the unique variables, we can write expressions for each of the unique variables in $f'$ and $g'$ as products of elements of the monomial group. We do the change of variable that takes $a_i$ to the $i$-th monomial of $f$ and $b_i$ to the $i$-th monomial of $g$ for the monomials that appear in these equations, yielding $d$ expressions of the form
\begin{align}
    a_i = m^{(f)}_i , \qquad
    b_j = m^{(g)}_j .
\end{align}
We can use these expressions to impose twisted boundary conditions on the remaining monomials that do not appear in these expressions. Therefore using these twisted conditions, we can rewrite the two polynomials $f$ and $g$ in terms of $v$ total variables. 

We conclude this section with an example that illustrates how to add twisted boundary conditions and do a change of variable to go from a high-dimensional HGP to a specific compactified model. 

\begin{example}
    Let $k_1 = 3$ and $k_2 = 2$. Then we begin with the product of the classical codes generated by \[1 + a_1 + a_2 + a_3 \qquad 1 + b_1 + b_2.\]

 Suppose we want to impose twisted boundary conditions on our classical code to realize the following fracton model:
 \[X \begin{pmatrix} 1 + xy + x^2y + y^3 \\ 1 + xz + z^2  \end{pmatrix} \quad Z \begin{pmatrix} \overline{1 + xz + z^2} \\ \overline{1 + xy + x^2y + y^3}  \end{pmatrix}.\]

We first check that the monomials of $f(x, y)$ and $g(x, z)$, along with their inverses and 1, form a group that contains $x, y, z$ respectively. We know that 
\begin{equation*}
   (xy)^{-1}(x^2y) = x, 
\end{equation*}
\begin{equation*}
    (xy)^2(x^2y)^{-1} = y,   
\end{equation*}
and 
\begin{equation*}
    (xy)(x^2y)^{-1}(xz) = z.   
\end{equation*}
Next, we do the change of variable that sends 
\begin{equation*}
    a_1 = xy \qquad a_2 = x^2y \qquad b_1 = xz.
\end{equation*}
Finally, we can use these expressions to impose the following twisted boundary conditions:
\begin{equation*}
    a_3 = (a_1^2a_2^{-1})^3 \qquad b_2 = (a_1a_2^{-1}b_1)^2.   
\end{equation*}
\end{example}

\subsubsection{Generalization to other translation-invariant qLDPC codes}

The above approach generalizes to an arbitrary infinite translation-invariant qLDPC stabilizer code with fixed check weights. This is because we can specify the stabilizers of these codes using the symplectic polynomial formalism as a matrix of generating polynomials, where each row corresponds to the $X$ or $Z$ support on a qubit position in the unit cell, and each column specifies the checks associated to the unit cell~\cite{haahThesis,toricStructureThm}.
We can then use the algebraic condition introduced above to check whether the code is indecomposable, by considering all of the terms across the polynomials in the matrix. The approach described previously generalizes directly to when there are more than two generating polynomials, by rewriting each $f(\mathbf{x})$ as $f = 1 + f'(\mathbf{x})$ and introducing a new variable for each term. We then use the code's indecomposability property to create the necessary twisted boundary conditions to recover the given qLDPC code from the high dimensional HGP parent.

\subsection{Fracton Family Trees}
\label{subsec:families}

When viewed as compactified fracton models, the A2BGA codes fall into three types of families over the ambient field $\mathbb{F}_2$. We call these \textit{fracton family trees}. Figure \ref{fig:tree} shows these families with some selected descendant codes, primarily taken from the classification paper Ref.~\cite{sorting}. The given HGP codes are the smallest representatives of a larger collection of HGP codes that unify the selected descendant models. The main difference of note between the different models is the parity of the length of the generating polynomials. This leads to the following proposition.

\begin{prop}
    Two A2BGA codes over \(\mathbb{F}_2\) are in the same fracton family tree if and only if the lengths of their generating polynomials have the same combination of parities.
\end{prop}

This splitting occurs since there is no way to change the parity of the lengths of the polynomials through compactification in $\mathbb{F}_2$. In the compactification process we can only eliminate terms from each of the parent polynomials by setting pairs of them to be equal. This means that we can never reduce the length of the generating polynomials by an odd amount, preserving the length's parity through the compactification process.

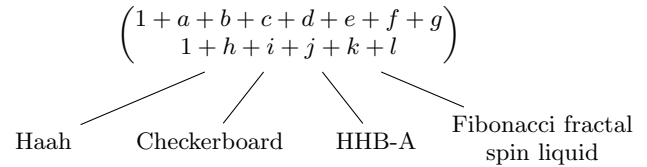
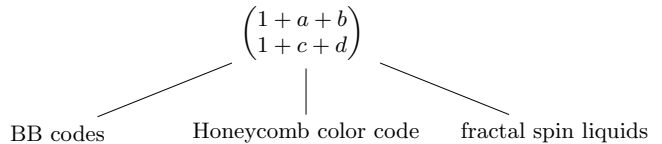
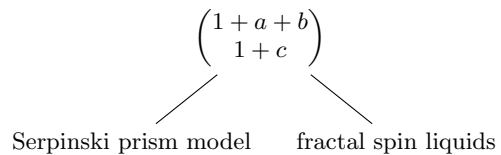
\begin{figure}[t]
\begin{subfigure}{\linewidth}
    \vspace{0.5cm}
    \begin{center}
        \begin{adjustbox}{max width=\columnwidth}
    \begin{tikzpicture}[
        level distance=1.4cm, % vertical spacing
        level 1/.style={sibling distance=2.2cm}, % horizontal spacing
        every node/.style={align=center}
    ]
        \node {
        $\begin{pmatrix}
            1 + a + b + c + d + e + f + g \\
            1 + h + i + j + k + l
        \end{pmatrix}$        
        }
            child {node {Haah}}
            child {node {Checkerboard}}
            child {node {HHB-A}}
            child {node {Fibonacci fractal\\spin liquid}};
    \end{tikzpicture}
    \end{adjustbox}
    \end{center}
    \vspace{-0.2cm}
    \caption{Family of codes with HGP parent generated by two even length polynomials.}
    \label{evens}
\end{subfigure}

\begin{subfigure}{\linewidth}
\vspace{0.5cm}
\begin{center}
    \begin{adjustbox}{max width=\columnwidth}
\begin{tikzpicture}[
    level distance=1.4cm, % vertical spacing
    level 1/.style={sibling distance=3.5cm}, % horizontal spacing
    every node/.style={align=center}
]
    \node {
    $\begin{pmatrix}
        1 + a + b \\
        1 + c + d
    \end{pmatrix}$        
    }
        child {node {BB codes}}
        child {node {Honeycomb color code}}
        child {node {fractal spin liquids}};
\end{tikzpicture}
\end{adjustbox}
\end{center}
\vspace{-0.2cm}
\caption{Family of codes with HGP parent generated by two odd length polynomials.}
\label{odds}
\end{subfigure}

\begin{subfigure}{\linewidth}
\vspace{0.5cm}
\begin{center}
    \begin{adjustbox}{max width=\columnwidth}
\begin{tikzpicture}[
    level distance=1.4cm, % vertical spacing
    level 1/.style={sibling distance=3.5cm}, % horizontal spacing
    every node/.style={align=center}
]
    \node {
    $\begin{pmatrix}
        1 + a + b \\
        1 + c
    \end{pmatrix}$        
    }
        child {node {Serpinski prism model}}
        child {node {fractal spin liquids}};
\end{tikzpicture}
\end{adjustbox}
\end{center}
\vspace{-0.2cm}
\caption{Family of codes with HGP parent generated by one even and one odd length polynomial.}
\label{evenodd}
\end{subfigure}

\caption{The three fracton family trees that arise in $\mathbb{F}_2$, defined by the parities of the lengths of the two generating polynomials of the HGP parent.
}
\label{fig:tree}
\end{figure}

Note that in Figure \ref{fig:tree}a we have one weight-eight polynomial and one weight-six polynomial. This is because the HHB-A model has generating polynomials 
    \begin{eqnarray}{l}
        f = 1 + x + y + z + xy + yz + xz + xyz, \\
        g = 1 + x^{-1}y + yz + xy + y^2 + yz^{-1}.
    \end{eqnarray}
In general, to find the explicit HGP generalization for codes in the same fracton family, we need generating polynomials with weights $w_1$ and $w_2$, where $w_1$ is the maximum number of terms in $f$ across all of the codes in the family, and $w_2$ is the maximum number of terms across the $g$ polynomials. To see the explicit twisted conditions needed to reduce the parent codes to the children shown in Figure \ref{fig:tree}, see Appendix \ref{appendix}.

% ////// DISTANCE BOUNDS /////////////////

\section{Code Distance Bounds}

\label{sec:Distance}

Next, we review results that provide upper and lower bounds on the distance scaling of the compactified codes. These bounds work generally for any A2BGA code, which includes many CSS fracton models and the BB codes.

In Ref.~\cite{BBdistBounds}, upper bounds for the distance scaling of a subset of A2BGA codes called generalized bicycle codes are derived. 
This was achieved by showing that they are equivalent to codes local in $D$ dimensions, and then using a generalized version of the Bravyi-Poulin-Terhal (BPT) bound \cite{BPTBound}.
We proceed similarly to derive the following upper bounds on the scaling of the compactifications. 

\begin{theorem}
\label{noDecouple}
    Suppose $\mathcal{C} = [[n, k, d]]$ is a family of abelian two-block group algebra codes defined by polynomials $f$ and $g$ with fixed weight $w$. Suppose that $\mathcal{C}$ is indecomposable, and that $v \leq w - 2$, where $v$ is the number of unique variables. Then we can map this code to one local in $D \leq w - 2$ dimensions, and the distance scaling of $\mathcal{C}$ has the following upper bound: \[d \leq O(n^{1 - (1/D)})\] with the parameters satisfying \[kd^{2/(D - 1)} \leq O(n).\]
\end{theorem}

\begin{proof}
    Section \ref{subsec:unification} details how to map any A2BGA code to a higher dimensional HGP parent code that is local in $w - 2$ dimensions. Since the child code has distance scaling upper bounded by the scaling of the parent code when considered on arbitrary twisted boundary conditions, we can apply the BPT bound to the HGP to achieve the desired result.
\end{proof}

If $\mathcal{C}$ is a family of \textit{local} A2BGA codes, we can apply the BPT bound directly to get $D$ equal to the number of unique variables that appear across $f$ and $g$ in Theorem \ref{noDecouple}.

In the theorem above, the \textit{weight} of a code written in the polynomial formalism is the total number of terms across the generating polynomials. If we have a family of decomposable codes defined on an infinite lattice, we can examine the distance bounds for each subcode separately. 
We can then map the indecomposable subcodes to HGP codes local in higher dimensions and apply Theorem \ref{noDecouple} to each subcode separately.

We remark that the above theorem applies similarly to more general translation-invariant qLDPC codes. This follow from the discussion at the end of subsection~\ref{subsec:unification}. 

We note that for a specific code in this family, if there are constraints on the periodic boundary conditions, it may not actually be decomposable. The theorem below gives an example of a specific instance of a code that is indecomposable, even if the model is decomposable when defined on an infinite lattice. 

\begin{theorem}
    If an A2BGA code $\mathcal{C}$ is defined with periodic boundary conditions of distinct prime lengths in each direction, then $\mathcal{C}$ is indecomposable.
\end{theorem}

\begin{proof}
    Suppose $\mathcal{C}$ is defined by polynomials $f$ and $g$. Then for each of the unique variables $x_i$, we know that $x_i^p = 1$ where $p$ is the prime length of the lattice in the $x_i$ direction. Suppose there exists a term $\alpha$ in the monomial group $G$ of $\mathcal{C}$ that contains the variable $x_i$. If $\alpha$ contains other variables in addition to $x_i$, we can raise $\alpha$ to the product of the distinct prime lengths for each of the other variables. Then we are left with a term only containing a power of $x_i$, so we reduce to the case where $\alpha = x_i^e$ for some power $e \nmid p$. Then since the exponents lie in $\mathbb{Z}/{p}$, there exists a multiplicative inverse $r$ such that $er \equiv 1$ mod($p$). Therefore $\alpha^{r} = x_i$, so $x_i$ lies in the monomial group. Since we can do this for each unique variable $x_i$, the monomial group is generated by the unique variables, as desired.
\end{proof}

Dua et al.~\cite{compactifyingFractons} and Chen et al.~\cite{yu-an} find lower bounds for the distance scaling of different A2BGA codes, where Dua et al.~look at three-dimensional fracton models, and where Chen et al.~consider BB codes. Both papers find these bounds by further compactifying A2BGA codes. Suppose we have a three-dimensional, or higher, translation-invariant stabilizer code that exhibits topological order, and has no string-like operators in the $z$-direction, and additional higher dimensional directions. Then, Dua et al.~found that this code can be compactified into a two-dimensional model. They find that this 2D model is unitarily equivalent to copies of the 2D toric code, using the following structure theorem.

\begin{theorem}[\cite{compactifyingFractons}, Section II A]
    \label{structure}
    Any two-dimensional translation-invariant CSS code with topological order over $R = \mathbb{F}_2[x^{\pm}, y^{\pm}]$ is equivalent to a tensor product of finitely many copies of the toric code and a product state.
\end{theorem}

Since these codes have distance that scales linearly with the lattice size, this implies that the original code also has distance that scales linearly with the lattice size. Our compactification process implies that a fracton model with no string-like operators in all but two directions is also equivalent to a tensor product of finitely many copies of the toric code. Then Theorem \ref{structure} gives us the following corollary.
This condition is satisfied for the simple HGP parent models of the A2BGA codes, provided the two directions that remain uncompactified are not wholly contained within the hyperplane of one input codes in the HGP parent model. 

\begin{corollary}
    \label{lowerBounds}
    If $\mathcal{C} = [[n, k, d]]$ is a $D$-dimensional fracton model with no string-like operators in all but two directions, then it has distance scaling lower bounded by \[O(n^{1/D}) \leq d.\]
\end{corollary}

However, if $\mathcal{C}$ has a string-like operator in one of the compactification directions, then this operator becomes local when we impose the twisted boundary conditions. This means that the distance of the code can collapse, so there is no definitive lower bound for this case. 

Chen et al.~also use Theorem \ref{structure} to find lower bounds on BB codes. Consider a family of BB codes with a fixed finite weight $f$ and $g$ and varying boundary conditions, such that the distance grows with the system size. Then this implies that the code has no local operators, and therefore is equivalent to copies of the toric code by Theorem \ref{structure}. This means that the distance scaling of BB codes is also lower bounded by $O(L)$ and therefore $O(\sqrt{n})$, as in Theorem \ref{lowerBounds}.

% ////// LOGICAL GATES /////////////////

\section{Other code properties}

In this section we discuss other properties of the parent and descendant codes and speculate about how they may be related under compactification. 

\subsection{Logical Gate Restriction Conjecture}

\label{sec:logicalGates}

We next propose a conjecture based on the fracton family tree structure. In particular, we conjecture that logical gate restrictions on the parent models are inherited by their descendants. Ref.~\cite{noGoHGP} derived a result for HGP codes showing that any logical transversal gate is restricted to the Clifford group. 
We now use this result to conjecture the same gate restriction for the compactified models. 
We remark that Ref.~\cite{Pastawski_2015} discussed transversal gate restrictions related to code compactification in type-II fracton models.

Suppose we have a code with stabilizer group $\mathcal{S}$. Then a (not necessarily Pauli-type) unitary operator $U$ is a \textit{logical operator} if it maps elements of $\mathcal{S}$ to stabilizers of the code. We can write this more formally as \cite[p. 23]{qecIntro}:
\begin{equation}
    USU^{\dagger} \in \mathcal{S} \quad \forall S \in \mathcal{S}
\end{equation}
A \textit{transversal gate} is a logical operator $U$ that can be written as the tensor product of unitaries acting on each qubit individually: 
\begin{equation}
    U = \bigotimes_{i = 1}^n u_i.
\end{equation}
We are interested in transversal logical gates because of their inherent fault-tolerance. Due to their decomposition of unitaries on each physical qubit, they do not couple the qubits within each code block, and therefore can only propagate errors between different code blocks.

We next define the \textit{Clifford hierarchy}. We define the first level of the Clifford hierarchy, $\mathcal{P}_0$, to be the multiplicative group containing $\pm I$. Next, we define $\mathcal{P}_1$ to be the Pauli group, which contains $I, X, Y, Z$ and each of these multiplied by $\pm i$. After this, the $i$-th level of the Clifford hierarchy is recursively defined as follows: 
\begin{equation}
   \mathcal{P}_i = \{U \mid UPU^{\dagger} \in \mathcal{P}_{i - 1} \quad \forall P \in \mathcal{P}_1\} 
\end{equation}
\cite{bravyiKonig}. We are particularly interested in $\mathcal{P}_2$, which is called the \textit{Clifford group}. This consists of all the unitary operators that send Pauli operators to Pauli operators. Crucially, if a code has transversal gates restricted to the Clifford group, then it is not possible to implement non-Clifford gates, such as the T-gate, transversally.

Below we state Fu et al.'s main result:
\begin{theorem}[\cite{noGoHGP}, Theorem 39]
    Let $Q$ be a $t$-dimensional ($t \geq 2$) hypergraph product code with $k$ logical qubits for any $k \geq 1$. If $U$ is a logical unitary implementable by transversal gates and the distance of the code satisfies $d \geq 3$, then $U$ is restricted to the Clifford group $\mathcal{P}_2$.
\end{theorem}

This immediately implies that all of our parent HGP codes have logical transversal gates restricted to the Clifford group. This leads us to make the following conjecture.
\begin{conjecture}
    Any A2BGA code has transversal gates restricted to the Clifford group.
\end{conjecture}

This would be a significant result, as it would imply that no A2BGA code can have non-Clifford transversal gates. We believe this to be true because the logical unitaries are the same from the parent code to the child code if the compactification radius is sufficiently large. However, the compactification process may create new logical operators not present in the parent code, so the set of logical operators of the child code may not be a subset of those of the parent. 

% ////// ENERGY BARRIER  /////////////////

\subsection{Energy Barrier Conjecture}

\label{sec:Energy}

Another application of our unification is in analyzing the energy barrier of the compactified codes, using results describing the energy barrier of the parent HGP code. We define the energy barrier and discuss the energy barrier of a HGP code following Ref.~\cite{hgpEnergyBarrier}.

Suppose that we have a stabilizer code $\mathcal{C}$ with parity-check matrix $H$, written in the binary symplectic form (see Haah's thesis for a review of this form \cite{haahThesis}). Let $v(P)$ be the bit-string representation of a Pauli operator $P$ when also written in this form. Then define the \textit{energy} of a Pauli operator $P$ by 
\begin{equation}
  \epsilon(P) = \text{wt}(H v(P)).  
\end{equation}
A sequence of Pauli operators $P_1,...,P_n$ such that $P_i$ and $P_{i + 1}$ differ by at most one qubit for all $1 \leq i < n$ is called a \textit{path}. Then the energy barrier of a path is given by the maximum energy of all the Pauli's in the path. The energy barrier of $P$, denoted by $\Delta(P)$, is the minimum energy barrier across all possible paths from $0$ to $P$. The energy barrier of $\mathcal{C}$ is then the minimum energy barrier over all of the nontrivial logical operators. The main result of Zhao et al.\@ on the energy barrier of a HGP is stated below.

\begin{theorem}[\cite{hgpEnergyBarrier}, Theorem 2]
    Suppose we have a HGP code with parity-check matrix $H = \begin{psmallmatrix}
        0 & H_Z \\
        H_X & 0
    \end{psmallmatrix}$.
    Also suppose that $H_X$ and $H_Z$ have constant row and column weight. If $H_X$, $H_Z$ and their transposes each have energy barrier that grows with the system size, then \[\Delta(H) = \textnormal{min}(\Delta(H_X), \Delta(H_Z), \Delta(H_X^T), \Delta(H_Z^T)).\]
\end{theorem}

Note that if $H_X$, $H_Z$ and their transposes each have a constant energy barrier, then the overall energy barrier is bounded by a constant. 
We expect that this will upper bound the possible energy barrier of the compactified codes. This is because the local scaling of the energy barrier for logical operator segments is equivalent in a parent code and its descendants. 
\begin{conjecture}
    \label{energy}
    The energy barrier scaling of any A2BGA code family is upper bounded by the energy barrier scaling of its HGP parent code family.
\end{conjecture}

It is known that the Newman-Moore model generated by $f(x, y) = 1 + x + y$ has logarithmic energy barrier \cite{Newman-MooreEnergyBarrier}. 
Hence, if Conjecture \ref{energy} is true, then all of the codes in the fracton tree generated by two copies of the Newman-Moore code, including weight six BB codes, would have energy barrier that grows at most logarithmically.

% ////// FUTURE WORK /////////////////

\section{Discussion}

\label{sec:Discussion}

We introduce a construction of a wide variety of translation invariant qLDPC codes from local fracton models in higher dimensions. For Abelian two-block group algebra (A2BGA) codes, including many fracton, fractal spin liquid, and bivariate-bicycle (BB) codes, the higher dimensional parent models are simple hypergraph-product (HGP) codes. We considered three fracton families defined on qubits, and we classified various known codes into these trees. 

We use our mapping of A2BGA codes to higher dimensional local codes to generalize upper distance bounds introduced for generalized bicycle codes in Ref.~\cite{BBdistBounds}. 
Our mapping motivated us to propose conjectures about the connection between the transversal logic gate restrictions and energy barriers of the parent and child codes. These conjectures are based on the fact that the local structure, including logical operator segments and their energy barriers, are the same for the child and parent codes provided there is a sufficiently large compactification radius.

There are a number of extensions to this project that would be interesting to explore. When examining the logical gate restrictions on codes, we only consider transversal gates. This could later be extended to search for restrictions on constant-depth circuits for the compactified models. We could also consider logic gates implemented via qLDPC code surgery procedures that use flexible ancilla systems~\cite{Williamson_2026,Ide_2025,yuan2026parsimoniousquantumlowdensityparitycheck} where it would be interesting to apply the structure of the higher dimensional fracton parent to find a systematic construction of performant ancilla systems.

The constructions of fracton family trees considered in this work are for models based on qubits. A distinct tree structure should emerge for models based on qudits~\cite{Kim_qudit_fracton}. Another potential direction is to explore whether our mapping informs decoding efficiency for the non-local child codes following a similar method to Refs.~\cite{BBDecoding, tan2026generalizedmatchingdecoders}. 
Another question is how coupled layer constructions of qLDPC codes~\cite{coupled_layer_product, rakovszky2023, rakovszky2024} are related to our higher dimensional mapping and coupled layer constructions of type-II fracton models~\cite{Hsieh_2017,Williamson_2021}. 
Similarly, one could ask how gapped boundary conditions of parent and child models are related.  
In this direction, Tile codes were introduced in Ref.~\cite{TileCodes} which equip the BB codes with a surface-code-like boundary following the standard procedure for introducing $X$ or $Z$ type boundaries to CSS codes, cf. Ref.~\cite{defectsCubic}. 
Finally, it would be interesting to establish a precise connection between the bifurcating entanglement-renormalization of the parent fracton models and the structure of their descendent codes~\cite{Haah_2014, Dua_2020}. 

\medskip 

\textbf{Note added.} While this work was being written Ref.~\cite{Shaw2026} appeared whose appendix also contains a discussion of how A2BGA codes can be viewed as higher dimensional local codes. 

\medskip 

\section*{Acknowledgements}
C.H.~would like to thank their fellow quantum REU students for valuable discussions.
The authors would like to thank William Gasarch for organizing the exceptional REU-CAAR program, during which this work was completed. 
DJW is supported by the Australian Research Council Discovery Early Career Research Award (DE220100625). 

\medskip 

\appendix

\section{Fracton Family Tree Compactifications}
\label{appendix}
Table \ref{fracton_tree_table} describes explicit twisted boundary conditions and change of variables necessary to place selected codes into one of the three fracton family trees described in Section \ref{subsec:families}. For the families of codes we unify into our trees, such as the BB codes and fractal spin liquids, we choose specific codes in order to write explicit twisted conditions. For the BB codes we decide to highlight the Gross code, as it is one of the most promising BB codes for near-term quantum error correction \cite{BBCodes}. Since all BB codes have the same weight, they are all children of the same parent code, though the specific twisted conditions necessary to recover them vary. Since the fractal spin liquids can have any combination of parities for the lengths of their generating polynomials, we highlight one example code for each case. Lastly, we note that our families only include child codes that are indecomposable. For instance, the HH Type-I code is not included as it is equivalent to two decoupled copies of the checkerboard model \cite{sorting}.

\onecolumngrid

\begin{table}[ht]
    \centering
    \caption{Explicit compactifications of parent HGP codes to selected child codes.}
    \begin{ruledtabular}
    \begin{tabular}{|c | c c|}
    Parent HGP & Child Codes & Twisted Conditions \\
    \hline \\[0.1mm]
    \multirow{5}{*}{$\begin{pmatrix}
            1 + a + b + c + d + e + f + g \\
            1 + h + i + j + k + l
        \end{pmatrix}$} & Haah: $\begin{pmatrix}
    1 + x + y + z \\
    1 + xy + xz + yz
\end{pmatrix}$ & $\begin{aligned}
a &= x, & b &= y, & c &= z \\
d &= a, & e &= a, & f &= a \\
g &= a, & h &= ab, & i &= ac \\
j &= bc, & k &= a, & l &= a
\end{aligned}$ \\[0.1mm] \\\cline{2-3} \\[0.1mm]
    & Checkerboard: $\begin{pmatrix}
    1 + x + y + z \\
    1 + x^{-1} + y^{-1} + z^{-1}
\end{pmatrix}$ & $\begin{aligned}
a &= x, & b &= y, & c &= z \\
d &= a, & e &= a, & f &= a \\
g &= a, & h &= a^{-1}, & i &= b^{-1} \\
j &= c^{-1}, & k &= a, & l &= a
\end{aligned}$ \\[0.1mm] \\\cline{2-3} \\[0.1mm]
    & HHB-A: $\begin{pmatrix}
    1 + x + y + z + xy + yz + xz + xyz \\
    1 + x^{-1}y + yz + xy + y^2 + yz^{-1}
\end{pmatrix}$ & $\begin{aligned}
a &= x, & b &= y, & c &= z \\
d &= ab, & e &= bc, & f &= ac \\
g &= abc, & h &= a^{-1}b, & i &= bc \\
j &= ac, & k &= b^2, & l &= bc^{-1}
\end{aligned}$  \\[0.1mm] \\\cline{2-3} \\[0.1mm]
    & fractal spin liquid (Fibonacci): $\begin{pmatrix}
    1 + z \\
    1 + x + x^{-1} + xy
\end{pmatrix}$ & $\begin{aligned}
a &= z, & b &= y, & c &= b \\
d &= b, & e &= b, & f &= b \\
g &= b, & h &= x, & i &= h^{-1} \\
j &= hb, & k &= b, & l &= b
\end{aligned}$ \\ \\[0.1mm]
    \hline \\[0.2mm]
    \multirow{3}{*}{$\begin{pmatrix}
            1 + a + b \\
            1 + c + d
        \end{pmatrix}$} 
        & BB code (Gross code): $\begin{pmatrix}
        1 + x^{3}y^{-1} + y \\
        1 + x^{-1}y^{3} + x
    
\end{pmatrix}$ & $\begin{aligned}
a &= b^{3}d^{-1}, & b &= y \\
c &= b^{-1}d^{3} & d &= x
\end{aligned}$ \\[0.1mm] \\\cline{2-3} \\[0.1mm]
    & Honeycomb color code: $\begin{pmatrix}
    1 + x + y \\
    1 + x^{-1} + y^{-1}
    
\end{pmatrix}$ & $\begin{aligned}
a &= x, & b &= y \\
c &= a^{-1} & d &= b^{-1}
\end{aligned}$  \\[0.1mm] \\\cline{2-3} \\[0.1mm]
    & fractal spin liquid example: $\begin{pmatrix}
    1 + (1 + x)y \\
    1 + (1 + x)z
\end{pmatrix}$ & $\begin{aligned}
a &= y, & b &= xy \\
c &= z & d &= a^{-1}bc
\end{aligned}$ \\ \\[0.2mm]
    \hline \\[0.2mm]
    \multirow{2}{*}{$\begin{pmatrix}
            1 + a + b \\
            1 + c
        \end{pmatrix}$} 
        & Serpinski prism model $\begin{pmatrix}
        1 + x + y \\
        1 + z
\end{pmatrix}$ & $\begin{aligned}
a &= x, & b &= y \\
c &= z
\end{aligned}$ \\[0.1mm] \\\cline{2-3} \\[0.1mm]
    & fractal spin liquid example $\begin{pmatrix}
    1 + (1 + x)y \\
    1 + z
    
\end{pmatrix}$ & $\begin{aligned}
a &= y, & b &= xy \\
c &= z &
\end{aligned}$ \\ \\[0.2mm]
    \end{tabular}
    \end{ruledtabular}
    \label{fracton_tree_table}
\end{table}

\bibliography{refs}

\end{document}